%% file: main.tex
\documentclass[10pt,letterpaper]{article}
\input{preamble}
\title{\papertitle}
\input{authors}
\date{}
\begin{document}
\hypersetup{pageanchor=false,pdftitle={Report the tolerance, not the count: exact phylogenetic distance deconvolution with incomplete references},pdfauthor={Haotian Tang}}
\begin{titlepage}
\maketitle
\thispagestyle{empty}
\begin{abstract}
\normalsize
\input{abstract}
\end{abstract}
\end{titlepage}
\hypersetup{pageanchor=true}
\setcounter{page}{1}
\input{sections/introduction}

\input{sections/model_results}
\input{sections/exact_algorithm}

\input{sections/partial_observations}

\input{sections/complexity_boundary}

\input{sections/experiments}
\input{sections/discussion}
\clearpage
\bibliographystyle{plain}
\bibliography{references}
\end{document}

%% file: preamble.tex
\usepackage[letterpaper,margin=1in]{geometry}
\usepackage{amsmath,amssymb,amsthm,booktabs,array,graphicx}
\usepackage[T1]{fontenc}
\usepackage{lmodern,microtype}
\usepackage[colorlinks=true,linkcolor=blue,citecolor=blue,urlcolor=blue]{hyperref}
\newtheorem{theorem}{Theorem}[section]
\newtheorem{lemma}[theorem]{Lemma}

\newtheorem{corollary}[theorem]{Corollary}
\theoremstyle{definition}

\theoremstyle{remark}

\newcommand{\dist}{\operatorname{dist}}

\newcommand{\LP}{\mathsf L}
\newcommand{\one}{\mathbf 1}

\newcommand{\bits}{\mathcal B}
\graphicspath{{figures/}}
\newcommand{\papertitle}{Report the tolerance, not the count:\\exact phylogenetic distance deconvolution with incomplete references}

%% file: authors.tex
\author{Haotian Tang\\[0.3em]
\normalsize Mohamed bin Zayed University of Artificial Intelligence\\
\normalsize\href{mailto:haotian.tang@mbzuai.ac.ae}{\texttt{haotian.tang@mbzuai.ac.ae}}}

%% file: abstract.tex
Phylogenetic distance deconvolution represents mixture-to-reference distances by weighted sources
placed on a reference tree. We study the minimum number of strictly interior sources that explains a
distance profile, and what incomplete references cost.

On a positive-rational reduced tree we give an exact algorithm with $2^{O(K)}\operatorname{poly}(\mathcal B)$
bit complexity for a free positive mean pendant and any nonempty set of observed leaves; the
single-exponential dependence on $K$ is tight under ETH, already with every leaf observed. With arbitrary
missing references and a prescribed pendant the problem is W[1]-hard in the source budget $K$, even on
positive binary trees supplied with a strict explanation using at most $K+1$ sources.

We then use the algorithm as an instrument. Removing references can only lower the minimum source count,
never raise it, at every error tolerance. Which references are missing then matters far more than how
many. Under nearest-to-source deletion a median of two
references can be erased before the generating support stops being the unique smallest explanation, and
that budget stays nearly flat as the panel grows (Spearman $0.10$ over 20 to 361 references, against
$0.82$ for random loss).
Discarding 7,931 of 10,575 references from a real 16S panel at random costs 1.8\% of the tolerance at which
one source suffices. Discarding comparably many as whole clades costs 90\%.

That tolerance needs neither ground truth nor an error model, which makes the model measurable on real data.
The generating configuration's worst-leaf residual exceeds the one-source optimum in all 296 real
ANI cases, and in 248 after refitting the common pendant. The phenomenon
is not unobservable in practice. A point estimate of the count alone does not quantify it. Methods for this
problem should report the tolerance at which an explanation of a given size exists.

%% file: sections/introduction.tex
\section{Introduction}
A sequence mixture contains contributions from several organisms, potentially
at evolutionary distances from the available references. Phylogenetic distance
deconvolution (PDD) models a distance from the whole mixture to each reference
as a weighted average of source-to-reference tree distances. Its sources may
lie inside reference-tree edges, allowing representation of constituents that
do not coincide with a reference genome. Arasti et al. introduced this
formulation, fixed-support fitting, exhaustive support search, and the practical
DecoDiPhy heuristic~\cite{pdd2026}.

We ask an exact computational question: does a given distance profile admit
an explanation with at most $K$ sources, and can a smallest such explanation
be constructed? The source edges, weights, and interior positions are unknown.
Enumerating edge subsets incurs a dependence of the form $N^{O(K)}$ on tree
size $N$. The question is whether the tree structure permits an algorithm
whose exponential dependence is confined to $K$. Missing reference distances
introduce an additional issue: projection onto the observed subtree changes
the common additive distance, or \emph{mean pendant}. Treating this quantity
as free or prescribed leads to different problems.

Our results give a complexity classification of these exact models. A
single-exponential algorithm applies to arbitrary missing leaves with a free
positive pendant, and to complete or one-missing-leaf observations with a
prescribed pendant. Complete-observation instances already admit restricted
NP/ETH lower bounds. With arbitrary missing leaves and a prescribed pendant,
the problem is W[1]-hard in the source budget even on strictly binary trees.
The latter reduction supplies an explicit legal explanation using at most one
source more than the requested budget: the difficulty persists when the data
are known to be exactly explainable.

The specialized ingredients are a bound on the branching excess of the
unknown distance potential above its leaf envelope, a strict source-count
preserving lift through omitted subtrees, and continuous-tree reductions
that rule out arbitrary competing placements. Established enumeration,
matching, and rational linear programming complete these arguments.

Exact parsimony is different from identifying the organisms that generated a
sample. We make no statistical consistency claim for an estimated distance
vector. Evidence below comes from exact profiles, where a missing entry is a removed
coordinate, and from published estimated distances, where we report the
tolerance admitting an explanation of a given size. Neither establishes
biological composition recovery.

\input{sections/related_work}

%% file: sections/related_work.tex
\section{Related work}\label{sec:related-work}
Phylogenetic placement puts a single query on a fixed reference tree, from an alignment and a likelihood model~\cite{pplacer2010,epang2019} or from estimated distances alone~\cite{apples2020}, with extensions to fragmentary markers and very large trees~\cite{sepp2012,bscampp2025}; $k$-mer methods supply the read-to-reference distances such placements consume~\cite{krepp2026}.
A placement is a one-source answer: one edge, and no support to choose.

Balaban and Mirarab's MISA method jointly decomposes mixture distances and places two constituents on a reference phylogeny, using a model derived from $k$-mer Jaccard indices~\cite{misa2020}.
Arasti et al. formulate PDD for arbitrary source counts and weighted additive mixture distances~\cite{pdd2026}.
Their work supplies same-edge aggregation, the reduction of pendant lengths to their weighted mean, fixed-support linearization and convex fitting, exhaustive support enumeration, and the DecoDiPhy heuristic.
It also exhibits nonidentifiability and studies restricted identifiable configurations, reporting that its failures need exactly equal branch lengths and particular proportions, and that identifiability is not a practical obstacle to its search.
We ask a different question of the same model: whether the minimum-count explanation is unique, and whether the parameters are pinned once a support is fixed. Neither holds as often as one would want, and the failures we measure sit on published reference phylogenies, not on symmetric constructions.
Arasti et al. name identifiability and robustness to distance error among the questions their framework leaves open; this paper takes up both, for the exact minimum source count with strict edge-interior placements and positive weights. One consequence of our deficit bound settles a case they leave open: on a reduced tree with all leaves observed, the occupied edges of an at-most-two-source explanation are determined even when the weights are unequal and unknown (Supplementary Appendix~E). The contribution is a source-budget complexity classification and what it measures, not the PDD model or fixed-support optimization.

Deconvolution against a phylogeny has been proved hard before, in cancer genomics. Hujdurovi\'{c} et al. split the rows of a binary
mutation matrix into the fewest rows admitting a perfect phylogeny; they correct an earlier hardness
proof, obtain NP-hardness by a different argument, and solve a restricted class of instances in
polynomial time~\cite{hujdurovic2018}. Zaccaria et al. pose copy-number tree mixture deconvolution
for multi-sample bulk sequencing tumor data~\cite{zaccaria2017}. Qi, Pradhan, and El-Kebir prove
that counting the solutions of perfect phylogeny mixture is \#P-complete and hard to approximate
within a constant factor~\cite{qi2019}. There the tree is unknown
and the mixing happens at its leaves over discrete states; here the topology and the edge lengths
are given, the sources sit at continuously chosen interior points of original edges, and our lower bounds
must rule out arbitrary competing continuous placements.

Sparse nonnegative fitting is the broader context. Bereg and Haghpanah prove NP-hardness for minimum-cardinality convex representations over arbitrary finite point sets when dimension is part of the input~\cite{bereghaghpanah2020}; their dictionary is finite, ours continuous.
Nadisic et al. give a branch-and-bound method for exact sparse nonnegative least squares and, biobjectively, the whole error-cardinality Pareto front instead of one chosen $k$~\cite{nadisic2020,nadisic2021}; that error-cardinality viewpoint is prior, and our contribution is the exact source-budget complexity classification for strict PDD; Damaschke and Sumita et al. obtain fixed-parameter and parameterized algorithms under sparsity restrictions on the system itself~\cite{damaschke2013,sumita2015}.
Neither transfers: our localized potential equations carry signed Laplacian coefficients without complementarity, and the sparsity we exploit comes from the geometry of a small source hull, not from an assumed bound on the input matrix.

Continuous-dictionary deconvolution is the nearest framework to our fixed-support linearization.
Duval and Peyr\'{e} show that above a signal-to-noise threshold, total-variation regularization over
measures recovers the exact number of Diracs of a nondegenerate sum, with a dual certificate
governing the limit as a discretizing grid is refined~\cite{duval2015}. Our source measure has unit total
mass, so a plain total-variation penalty adds a constant and prefers neither fewer sources nor
fewer occupied edges; the sparsity minimized here counts occupied edges, and interpolation along an
edge is exact, not a grid approximation. Their certificate buys a recovered count under nondegeneracy, whereas here an
arbitrarily small change can raise the minimum count to any prescribed value (Supplementary
Appendix~D). A recent general treatment of cardinality-bounded nonnegative approximation contains our feasibility questions as instances, not the reverse~\cite{schachter2026}. Bienstock and Mu\~{n}oz approximate polynomial optimization by polynomial-size linear programs at bounded tree-width~\cite{bienstock2018}, which does not settle an exact rational feasibility question.

Several ingredients of our proofs are classical.
The weighted-tree distance inverse is due to Bapat, Kirkland, and Neumann~\cite{bapat2005}, and the lower Lipschitz envelope is a standard extension construction~\cite{mcshane1934}.
Directional mass formulas for metric-tree distance potentials also appear in the theory of Fr\'echet medians~\cite{romon2023}.
Our anchored search uses the include-or-exclude principle underlying connected-set enumeration~\cite{fomin2012}.
What is new is the geometry that bounds and locates the envelope deficits. Oettli and Prager characterized by signed inequalities the solvability of a linear system known only within bounds~\cite{oettli1964}; we state everything for exact data.
Cardinality-constrained sum encodings and parameterized sum problems have precedents~\cite{defarias2003,abboud2014}; the reductions here must additionally enforce normalized positive mixtures, continuous placements, and the stated observation and pendant constraints.

%% file: sections/model_results.tex
\section{Model and main results}\label{sec:model}
Let $T=(V,E)$ be a finite \emph{reduced} tree: every nonleaf vertex has degree
at least three. Its edge lengths are positive rationals, $N=|V|\ge2$, and
$\varnothing\ne R$ is a subset of the original leaves. The supplied distances
are $d_r\in\mathbb Q$ for $r\in R$. A \emph{strict explanation} has $k\ge1$
points $q_i$ in the interiors of distinct original edges, weights $p_i>0$
with $\sum_i p_i=1$, and a positive mean pendant $y$, satisfying
\begin{equation}\label{eq:model}
 d_r=y+\sum_{i=1}^k p_i\dist_T(r,q_i),\qquad r\in R.
\end{equation}
If the individual pendant lengths are $y_i$, only their weighted mean
$y=\sum_i p_i y_i$ enters these observations. The model and this ambiguity
are inherited from PDD~\cite{pdd2026}. Sources on one edge can be merged at
their weighted position without altering original-vertex reports or losing
strictness. We therefore count distinct source edges (Figure~\ref{fig:model}, left).

The \emph{source-budget problem} decides whether $k\le K$, and returns the
minimum count within the budget and a strict witness when one exists. We
distinguish a free positive $y$ from a prescribed rational $y_0>0$.
The latter is an additional information assumption, not a quantity presumed
known in every biological application. Write $\bits$ for the total binary
encoding length of the rational input, including any supplied generator.
We may clamp $K$ to $|E|$; budget zero is infeasible.

\begin{theorem}[Exact upper bounds]\label{thm:upper}
On a reduced positive-rational tree, with any nonempty observed-leaf subset,
the free-positive-pendant problem has a deterministic algorithm using
$O(N)+2^{O(K)}$ rational arithmetic operations and
$2^{O(K)}\operatorname{poly}(\bits)$ bit time. It returns the minimum count
within $K$ and a polynomial-bit rational strict witness, or reports
infeasibility within $K$. The same bounds hold with prescribed positive
pendant when at most one original leaf is unobserved.
\end{theorem}

\begin{theorem}[Restricted lower bounds]\label{thm:hardness}
The source-budget problem is NP-complete on positive-rational binary trees
with mean pendant prescribed to be one, even with all leaves observed and an
explicit larger legal strict generating mixture supplied as input. The same
holds with exactly one missing leaf. Under deterministic ETH, neither
restriction admits a deterministic
$2^{o(K)}\operatorname{poly}(\bits)$ algorithm.
\end{theorem}

The supplied-generator version is a decision language: invalid generators
are rejected by exact polynomial-time checking. Every instance produced by
our reductions has a valid generator on both YES and NO inputs.
The lower bound matches the linear order of the exponent in $K$, not an
optimal exponential base. It is not strong NP-hardness.

\paragraph{Lower-bound mechanism.}
Theorem~\ref{thm:hardness} is proved in Supplementary Appendix A by reducing
cardinality-constrained \textsc{Subset Sum}. Each integer is represented by a
binary cherry: an inactive item costs at least one source edge, while an
active item costs two. A calibration edge is mandatory. Normalization and
nonnegative pendant length impose two resource inequalities; together with
the source budget, they force an exact integer sum. Explicit positive core
lengths preserve this gap for arbitrary continuous competing placements.
Strict YES witnesses and a larger answer-independent generator share one
fixed pendant, which is rescaled to one. The NO argument there excludes every
nonnegative pendant, so both lower bounds apply to the free-mean row too. The ETH transfer uses a carry-free
encoding of sparse 3-SAT with source budget linear in the formula size.

\begin{table}[!htbp]
\centering
\caption{Exact source-budget complexity. All sources lie strictly inside
distinct original edges. The upper bounds hold on arbitrary reduced trees;
the lower bounds already hold on positive binary trees.}
\label{tab:classification}
\begin{tabular}{@{}p{0.27\linewidth}p{0.20\linewidth}p{0.45\linewidth}@{}}
\toprule
Observed leaves & Mean pendant & Result \\
\midrule
Any nonempty subset & Free $y>0$ & $2^{O(K)}\operatorname{poly}(\bits)$ time; NP-complete and tight under ETH already with every leaf observed \\
All, or all but one & Prescribed $y_0>0$ & The same bounds \\
Arbitrary subset & Prescribed $y_0=1$ & W[1]-hard in $K$, even with a supplied strict $(K+1)$-source explanation \\
\bottomrule
\end{tabular}
\end{table}

The third row is proved in Theorem~\ref{thm:partial-hardness}. Its contrast
with the first two rows is a statement about different information regimes,
not a contradiction between the upper and lower bounds.

\begin{figure}[t]
\centering
\includegraphics[width=\linewidth]{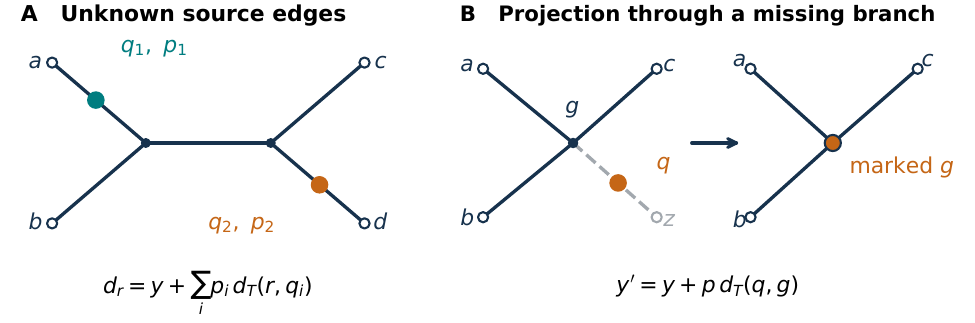}
\caption{The exact model and the role of omitted references. Left: two
weighted sources on unknown original edges generate mixture-to-leaf distances.
Right: a source on an omitted branch projects to its attachment gate. Its
weighted distance to the observed subtree is absorbed into the mean pendant.
A retained branch gate must permit a vertex atom; deleting that permission
can change the minimum source count. The drawings are schematic, not to scale.}
\label{fig:model}
\end{figure}

%% file: sections/exact_algorithm.tex
\section{Tree geometry and localization}\label{sec:geometry}
For now let $R$ contain every leaf. Let $D$ be the full vertex distance matrix, $\LP$ the Laplacian with conductance $1/\ell_e$ on edge $e$, $\tau_v=2-\deg(v)$ and $W=\sum_e\ell_e$. The weighted-tree inverse identity gives \cite{bapat2005}
\begin{equation}\label{eq:inverse-identities}
 D\tau=W\one,\qquad D\LP=\one\tau^\top-2I.
\end{equation}
An edge source at fraction $x$ from endpoint $u$ to $v$ gives endpoint masses $p(1-x)$ and $px$. Thus its vertex potential is $f=Dm+y\one$ for a nonnegative mass vector of sum one. Conversely, for any vertex potential $f$, set
\begin{equation}\label{eq:inverse}
 m=(\tau-\LP f)/2,\qquad y=(\tau^\top f-W)/2.
\end{equation}
Then $\one^\top m=1$ since $\one^\top\tau=2$ and $\one^\top\LP=0$, and \eqref{eq:inverse-identities} gives $Dm+y\one=f$. Nonnegative masses still need a legal strict source realization; they are not automatically valid strict edge sources.

Define the lower Lipschitz envelope and baseline
\begin{equation}\label{eq:envelope}
 h(v)=\max_{r\in R}\{d_r-\dist(v,r)\},\quad
 b=(\tau-\LP h)/2,\quad Y=(\tau^\top h-W)/2.
\end{equation}
The envelope is classical \cite{mcshane1934}. Because every feasible $f$ is 1-Lipschitz, $f\ge h$. At an observed leaf feasibility implies $h(r)=d_r$; otherwise reject. Put $\delta=f-h$ and $S=\{v:\delta_v>0\}$. Then $\delta$ is zero at all leaves.

\begin{lemma}[Weighted deficit bound]\label{lem:hull}
If at most $K$ sources realize $f$, then $S=\varnothing$ for $K=1$, and
also for $K=2$ when every nonleaf vertex has degree at least three. For $K\ge2$,
\begin{equation}\label{eq:weighted-budget}
 \sum_{v\in S}(\deg(v)-2)\le K-2.
\end{equation}
The assertions still hold when sources may sit at internal vertices, each
such atom counted in $K$.
\end{lemma}
A path identity makes $\delta_v$ positive exactly when every component of $T-v$ contains a source
away from $v$: if one direction carries no mass, an observed leaf in it reports $f$ exactly and the
envelope is already tight at $v$. Such a $v$ is a branch point of the minimal subtree $H$ spanning
the sources, with $H$-degree equal to its $T$-degree, and a tree with at most $K$ terminals
satisfies $\sum_v(\deg_H(v)-2)\le K-2$ over its branch points, which is
\eqref{eq:weighted-budget}. The unknown part of the potential is therefore not merely small, it
sits on the branching of a $K$-terminal subtree; and a boundary sum over each component of $T[S]$
leaves a strictly positive flux, giving every component a positive baseline mass, which is where
the enumeration starts. Supplementary Appendix~E proves both.

The inverse formulas give
\begin{equation}\label{eq:local-affine}
 m=b-\LP\delta/2,\qquad
 y=Y-\tfrac12\sum_{v\in S}(\deg(v)-2)\delta_v.
\end{equation}
Only $A=S\cup N_T(S)$ changes mass. For $K\ge3$, \eqref{eq:weighted-budget} implies
\begin{equation}\label{eq:local-size}
 |S|\le K-2,\quad \sum_{v\in S}\deg(v)\le3(K-2),\quad |A|\le4K-8.
\end{equation}

\begin{lemma}[Positive anchors and baseline filters]\label{lem:anchor}
Every connected component of $T[S]$ contains a vertex with $b_v>0$. For $K\ge2$, at most $3K-2$ baseline masses are positive. For $K=1$, at most two are positive. For $K\ge3$, at most $4K-8$ are negative; for $K\le2$ none are negative.
\end{lemma}
Summing $b$ over a component of $T[S]$ leaves a positive boundary term, and a $K$-source mixture
carries at most $2K$ positive endpoint masses. Supplementary Appendix~E completes the counts.

\section{The exact algorithm and encoding bounds}\label{sec:algorithm}
We prove the upper bound first for complete observations. Compute $h$ by a postorder maximum pass and a preorder pass propagating the best value outside each child subtree. Maintaining the two best incoming values at each vertex gives $O(N)$ arithmetic and comparison operations. Compute $b,Y$, degrees, a parent map and baseline support lists in another tree pass. Reject inconsistent leaf envelopes or any count filter of Lemma~\ref{lem:anchor}. Reject $Y\le0$ when the pendant is free, and $Y<y_0$ when it is prescribed to be $y_0$: the deficit term subtracted in \eqref{eq:local-affine} is nonnegative, so $y\le Y$ in either regime. The positive and negative baseline lists now both have $O(K)$ entries.

\subsection{Complete single-exponential support enumeration}
For $K\le2$ inspect only $S=\varnothing$. Otherwise a possible deficit vertex is internal with degree between three and $K$. Call these vertices eligible and assign weight $w(v)=\deg(v)-2$. The eligible vertices with $b_v>0$ are anchors; there are $a\le3K-2$ of them. Enumerate every anchor subset $Q$ whose total weight is at most $K-2$. It is interpreted as exactly the set of anchors in $S$; all other anchors are permanently excluded.

Initialize the included vertices to $Q$ and the frontier to eligible unclassified neighbors of included vertices. In a fixed vertex order, either exclude the next frontier vertex permanently, or include it if its weight fits the remaining budget. Only inclusion scans its neighbors. When the frontier is empty, emit the included set. The empty anchor subset emits the empty set.

For any target $S$ satisfying the weighted budget and component-anchor condition, choose $Q=S\cap\{\text{anchors}\}$ and include precisely its vertices whenever encountered. Every component is reached from its seed, and every prefix stays within budget. Thus all admissible supports are emitted. The fixed order and permanent decisions give a unique branch for each target and seed mask.

Along a search branch, each queried vertex has been exposed by an included-vertex incidence. Inclusion is permanent, so every vertex queried anywhere along a branch is incident to that branch's own final included set $U$. The number of distinct queried vertices on the branch is therefore at most $\sum_{v\in U}\deg(v)\le3(K-2)$. The binary search tree therefore has at most $2^{3(K-2)+1}-1$ states for one seed mask. Multiplication by at most $2^a$ masks gives $2^{O(K)}$ states. Their bookkeeping is polynomial in $K$. No rejected vertex's adjacency is scanned.

This is not a new general enumeration technique: joining all anchors to an artificial root makes every admissible support connected with $O(K)$ vertices and external neighbors, so the rooted connected-set enumeration of Fomin and Villanger applies \cite[Lemmas 3.1--3.2]{fomin2012}. Either route needs Lemmas~\ref{lem:hull} and~\ref{lem:anchor}.

\subsection{Strict realization from positive masses}
Let $P=\{v:m_v>0\}$. A strict edge atom places positive mass at both endpoints, so a realization is an edge cover of the induced forest $T[P]$. Such a cover exists exactly when $T[P]$ has no isolated vertices; its minimum cardinality is $|P|-\nu(T[P])$, where $\nu$ is maximum matching size.

For later use suppose vertex atoms are allowed at a marked set $M$. Then a cover exists exactly when every isolated vertex of $T[P]$ belongs to $M$, and its minimum is still $|P|-\nu(T[P])$. A maximum matching, one incident edge for every unmatched nonisolated vertex, and a singleton for each isolate attain this count. For a lower bound, remove redundant members from any pair/singleton cover. Every selected edge has a private endpoint. Each nontrivial component of selected edges is therefore a star, and selected singletons are disjoint from selected edges. If there are $s$ stars the count is $|P|-s$, while one edge per star is a matching, giving $s\le\nu(T[P])$.

To realize a cover, give each isolated singleton its full mass, and split every nonisolated vertex mass equally among selected incident edges. On a selected edge with endpoint contributions $a,b>0$, use source weight $a+b$ and fraction $b/(a+b)$ from the endpoint receiving $a$. Positions are strictly interior and the weights sum to one. This is classical covering machinery; no new matching theorem is required.

\subsection{Exact zero faces and the common slack}
For each emitted $S$, only $A=S\cup N(S)$ changes. Use the stored baseline lists to reject negative masses outside $A$. Enumerate each $Z\subseteq A$ as the exact zero-mass coordinates inside $A$, and let
\begin{equation}\label{eq:face-support}
 P=(A\setminus Z)\cup\{v\notin A:b_v>0\}.
\end{equation}
Reject $|P|>2K$ or an illegal cover or a cover exceeding $K$. Build the induced forest using the global parent map and membership in $P$, without scanning a high-degree vertex's full adjacency.

For each remaining face use variables $\delta_v$ for $v\in S$ and a slack $t$:
\begin{align}\label{eq:slack}
 0\le t\le1,&\qquad \delta_v\ge t\quad(v\in S),\notag\\
 m_v(\delta)=0\quad(v\in Z),&\qquad m_v(\delta)\ge t\quad(v\in A\setminus Z),\\
 &y(\delta)\ge t.\notag
\end{align}
Fixed positive masses outside $A$ are already strictly positive and need not be bounded below by $t$. A strict solution on the specified supports exists if and only if this polytope contains $t>0$.

The polytope is bounded: $0\le t\le1$, $\delta\ge0$ and the pendant inequality implies $\sum_{v\in S}(\deg(v)-2)\delta_v\le2Y$. Every coefficient on $S$ is positive. Thus a positive maximum $t$ is attained at a vertex. Equalities may be replaced by two inequalities. Enumerating full-rank active bases finds all needed vertices, including vertices of lower-dimensional polytopes: at any extreme point the active normals span the ambient space. There are $O(K)$ inequalities and at most $K$ variables, including the slack when $K=1$, so at most $2^{O(K)}$ bases. Retain a feasible positive-slack vertex, realize its minimum cover, and keep the smallest cover over all supports and faces. Testing only the relative interior of the whole nonnegative polytope would miss solutions on mass-zero faces; their enumeration is essential.

\subsection{What the enumeration computes}
Call a pair $(S,Z)$ \emph{emitted} when the support search above emits $S$ and $Z\subseteq A$. Call it \emph{slack-positive} when no $b_v$ with $v\notin A$ is negative and \eqref{eq:slack} has a solution with $t>0$. Write $\operatorname{cov}(S,Z)=|P|-\nu(T[P])$ for $P$ as in \eqref{eq:face-support}, and $\operatorname{cov}(S,Z)=\infty$ when an isolated vertex of $T[P]$ carries no permitted atom. Let $k^\star$ be the least number of sources in a strict explanation of the data, and $\infty$ when there is none.

\begin{lemma}[Characterization]\label{lem:characterization}
Over emitted slack-positive pairs, the minimum of $\operatorname{cov}$ equals $k^\star$ when $k^\star\le K$, and exceeds $K$ otherwise. An empty minimum is $\infty$.
\end{lemma}
The proof is a construction and its inverse. A slack-positive pair yields a strict explanation with
$\operatorname{cov}(S,Z)$ sources; a strict explanation with $k\le K$ sources yields an emitted
slack-positive pair with $\operatorname{cov}\le k$. Supplementary Appendix~E gives both directions.

The set $P$ is a function of $(S,Z)$ and not of the point chosen in the face, so every positive-slack point of a face has the same cover number. Retaining one feasible vertex is therefore sound, although its deficit vector need not be the one an optimal explanation realizes. The remaining filters drop only pairs that cannot pay for themselves: $|P|>2K$ forces $|P|-\nu(T[P])\ge|P|/2>K$.

For a prescribed positive pendant $y_0$ and complete observations, replace $y(\delta)\ge t$ in \eqref{eq:slack} by $y(\delta)=y_0$. The equality pins $\sum_{v\in S}(\deg(v)-2)\delta_v=2(Y-y_0)$, so the polytope stays bounded and the active-basis enumeration still reaches every vertex. Both halves of the proof persist with pendant exactly $y_0$. With $S=\varnothing$ the equality is the entire test.

\subsection{Arithmetic and bit complexity}
After the $O(N)$ preprocessing, supports, zero faces and active bases each number $2^{O(K)}$. Their product remains single-exponential; local operations are polynomial in $K$. Source realization and induced-forest matching operate on $O(K)$ vertices. This proves the rational arithmetic bound.

For the bit bound, every $h(v)$ is an input observation minus a path-length sum. A common denominator whose bit length is bounded by the total input size bounds all such values polynomially in $\bits$. Each $b_v$ is a sum of neighboring potential differences divided by input edge lengths; $Y$ is another finite weighted sum. Their bit lengths are polynomial in $\bits$. Each local active system has $O(K)$ rows with these coefficients and reciprocals of lengths. Clearing row denominators and applying determinant bounds gives polynomial-bit solutions, because $K\le|E|\le\bits$. Exact elimination with reduced fractions and all sign checks require polynomial bit time per basis. Positive splitting and source reconstruction preserve polynomial encoding length. Hence the full bit complexity is $2^{O(K)}\operatorname{poly}(\bits)$ and a strict witness has polynomial size. This is not a bound of $O(\bits)+2^{O(K)}$ bit operations or a polynomial-kernel claim.

%% file: sections/partial_observations.tex
\section{Partial observations and strict lifting}\label{sec:partial}
With a free positive mean pendant and $R=\{r\}$, feasibility is exactly $d_r>0$: use one source at distance $a=\min\{\ell/2,d_r/2\}$ from $r$ on its incident edge and pendant $d_r-a$. Now assume $|R|\ge2$.

Let $H$ be the minimal subtree spanning $R$, retaining original vertices. Suppress its degree-two vertices to obtain a reduced backbone $B$. Its leaves are exactly $R$. Mark
\[
 M=\{g\in V(B):\deg_H(g)\ge3,\ \deg_T(g)>\deg_H(g)\}.
\]
The reduced problem permits strict edge atoms and atoms at these marked branch vertices, with a free positive pendant. Atoms at suppressed degree-two gates are already strict interior points of backbone edges.

\begin{lemma}[Budget-preserving projection]\label{lem:projection}
Minimum strict original-edge source count on $T$ equals minimum mixed edge/marked-vertex atom count on $(B,M)$, including infeasibility.
\end{lemma}
Every point $q$ has a unique gate $g$ in $H$ with $\dist_T(r,q)=\dist_H(r,g)+\dist_T(q,H)$ for
every $r\in R$, which projects an explanation onto $(B,M)$ without raising its count. A free
pendant pays to lift a backbone atom back into a deleted incident edge, and the lifted original
edges are distinct. Supplementary Appendix~E proves both directions.

The inverse formulas, weighted deficit bound, anchor identity and baseline filters apply to mixed atoms. A vertex atom at $v$ contributes zero to every initial-path overlap from $v$, so the proof of Lemma~\ref{lem:hull} is unchanged. Every atom contributes at most two positive vertex masses. The only change in the algorithm is the isolated-vertex marked-cover test already proved above, and Lemma~\ref{lem:characterization} then holds verbatim on $(B,M)$ with singleton atoms permitted at the marked gates. Mark flags are stored once and inspected only for the $O(K)$ current positive vertices.

Construct $H$ by pruning unobserved leaves; construct $B$ by scanning each retained chain once, storing its original edge list and one deleted incident edge at each gate. These operations are linear. For lifting, scan only the chains bearing the selected distinct edge atoms; their total length is at most $N$. All backbone lengths are sums of original rational lengths and have polynomial bit size. Taking minima with $Y/2$ and adjusting the pendant preserves the polynomial-bit witness bound. This proves the free-pendant part of Theorem~\ref{thm:upper}.

The marks are necessary. In a four-leaf unit star, a midpoint source on the omitted arm with pendant one gives distance $5/2$ to each of three observed leaves. Its projection is a marked center atom of count one. If that mark is discarded, a strict explanation on the three-arm star needs three sources: with at most two sources $S$ is empty, and the inverse of the center profile has its only positive mass at the isolated unmarked center. Three equal sources at arm midpoints and pendant $4/3$ realize the profile. With $d$ observed arms the gap runs from one to $d$: the reduced star's only internal vertex has degree $d$, so Lemma~\ref{lem:hull} forces $S=\varnothing$ for every $k<d$, the isolated unmarked center again blocks every cover, and $d$ midpoint sources with pendant $(d+1)/d$ realize the profile. This is a source-budget equivalence, not an original-support identification theorem. If the original pendant is fixed, projecting omitted distances changes it; the general fixed-pendant partial-observation upper bound does not follow.

For exactly one omitted leaf $z$, however, retain its neighbor $c$ rather
than suppressing it. A hidden source contributes a mandatory atom at $c$.
If the deleted edge has length $\ell$, projected mass $m_c$, and projected
pendant $Y_p$, its exact feasibility condition is
$0<Y_p-y_0<\ell m_c$. A mandatory-singleton edge-cover calculation and
common-slack constraints enforce this condition. There is at most one
degree-two vertex, so the deficit search still has $2^{O(K)}$ states.
Supplementary Appendix C gives the full algorithm, including bounds for that
zero-branching-weight coordinate and strict witness reconstruction.

%% file: sections/complexity_boundary.tex
\section{A parameterized boundary for missing references}
\label{sec:complexity-boundary}

The free-mean algorithm permits an arbitrary nonempty subset of observed
original leaves. Prescribing the mean changes this conclusion: an unobserved
excursion can no longer be absorbed into an unconstrained additive term.
The resulting obstruction is parameterized by the number of sources, rather
than by the number of candidate edges. The problem stays in $\mathsf{XP}$ by
enumerating supports, so what the hardness below rules out is detaching the
exponent from $N$, not solvability.

\begin{theorem}[Missing references with prescribed mean]
\label{thm:partial-hardness}
The strict source-budget problem is $\mathsf{W[1]}$-hard parameterized by $K$
on positive-rational unrooted strictly binary trees when an arbitrary subset
of original leaves is observed and the common mean is prescribed as $y=1$.
Hardness holds even when a polynomial-bit strict explanation using at most
$K+1$ sources is supplied and the minimum is promised to be either $K$ or
$K+1$. On the NO instances of the reduction, no normalized mixture of at most
$K$ arbitrary closed tree points explains the observations, even if vertex
locations and repeated edges are allowed.
\end{theorem}

\begin{proof}[Proof idea]
A carry-free encoding reduces $t$-\textsc{Clique} to choosing one positive
integer from each of $g=t+\binom{t}{2}$ groups with a specified sum.
The tree has $g$ active arms, a calibration arm, and a sentinel reference.
Each integer $a$ on an active arm is represented by an unobserved unit twig
attached at coordinate $1/(1+a)$. Observed distance contrasts force one
source in every active arm and one on the calibration arm. Thus a budget
$K=g+1$ leaves no source elsewhere. The calibration contrast bounds the
total active mass from above; the prescribed mean and the sentinel distance
bound the hidden-distance moment from below. Together they exclude every
visible-arm source. All remaining sources select encoded integer gates,
and an integer-gap inequality forces their sum to equal the target.

A positive binary core replaces the high-degree junction. Its lengths are
chosen explicitly so that an exact rational linear system produces strict
YES witnesses and the integer-gap argument remains valid for every closed
competitor. Interpolating group minima and maxima requires splitting at most
one group and gives the supplied $K+1$-source explanation without solving the
selection problem. Supplementary Appendix~B proves the construction,
all strictness bounds, and polynomial encoding.
\end{proof}

\begin{corollary}[An upper bound on an unknown mean]
\label{cor:bounded-mean}
The same hardness and one-source promise gap hold when the mean is unknown
but constrained by $0<y\le1$. The supplied explanation has $y=1/2$; NO
instances exclude at most $K$ closed-point sources even for $0\le y\le1$.
\end{corollary}

Consequently, an algorithm with running time $f(K)\operatorname{poly}(\mathcal B)$
for either restricted model would imply $\mathsf{FPT}=\mathsf{W[1]}$.
This does not contradict the free-positive-mean upper bound: imposing an
external upper bound is a different feasibility constraint. The reductions
use potentially short edges and small source margins; they establish neither
bounded-aspect-ratio hardness nor robustness to observational error.

%% file: sections/experiments.tex
\section{What incomplete references cost}
\label{sec:experiments}
The algorithm is also an instrument. We ran it on a frozen bank of exactly explainable profiles
built from published reference phylogenies, 20 to 361 observed leaves, $K=2$--12, and on 385 real
distance profiles: 296 ANI and 60 placement profiles from the parent
benchmark~\cite{pdd2026,krepp2026} and 29 stool samples on a 10,575-leaf 16S backbone. For a support $S$ of $k$ distinct edges let $\theta(S)$ be the least worst-leaf deviation of
\eqref{eq:model} over its closed version, one exact rational linear program, and let
$\Theta_k=\min_{\lvert S\rvert=k}\theta(S)$, non-increasing in $k$. For $k=1$, two tree traversals and a
constant-size minimization per edge give $\Theta_1$ and a closed attaining witness in
$O(\lvert V\rvert)$ exact rational operations, which is what makes the diagnostic affordable at
10,575 references. The minimum
count within a tolerance $\varepsilon>0$ is $m(\varepsilon)=\min\{k:\Theta_k<\varepsilon\}$; the strict
inequality is a convention that matters only at $\varepsilon=\Theta_k>0$, where an explanation
attaining the bound exactly need not satisfy the strict positivity constraints. At
$\varepsilon=0$ the tolerant program returns the frozen exact solver's minimum on all 1,118 leaf sets
checked against it. Protocols, tables and further figures are in the supplementary experimental
details.

\paragraph{One-sidedness.}
The $2^{O(K)}\operatorname{poly}(\mathcal B)$ bound of Section~\ref{sec:algorithm} is for
$\varepsilon=0$: at positive tolerances we enumerate supports, and large-panel one-source analyses
use a separate closed form. Deleting an observed leaf deletes two rows from every one of these programs, so each $\Theta_k$ can
only fall, at any $\varepsilon\ge0$, and nothing rose in any verification draw (Supplementary
Appendix~F).
The exact minimum fell in 450 of 1,968 completed paired draws, by up to six sources; those draws
come from the 83 of 101 registered cases whose full panel completed. The fall
rate rises with the level, from 0.2\% for a single reference to 59.1\% at 75\% loss. Which references go matters more than how many: whole-clade removal
lowered the count in 32.5\% of draws and nearest-to-source removal in 29.7\%, against 10.3\% at
random, and those ratios hold between 2.0 and 4.0 at every non-vacuous tolerance
(Figure~\ref{fig:tolerance}b). The whole-taxon regime
tested here behaves like random removal on the ANI panels, a ratio of 1.1 at 75\% loss. Its granularity undershoots the nominal
level, a median 0.733 where 0.75 was asked on the ANI trees but only 0.386 on the placement trees,
so the two arms are not comparable at a fixed nominal level and we draw no placement conclusion from
it; panel breadth confines the ranks to genus and species.

\begin{figure}[t]
  \centering
  \includegraphics[width=\linewidth]{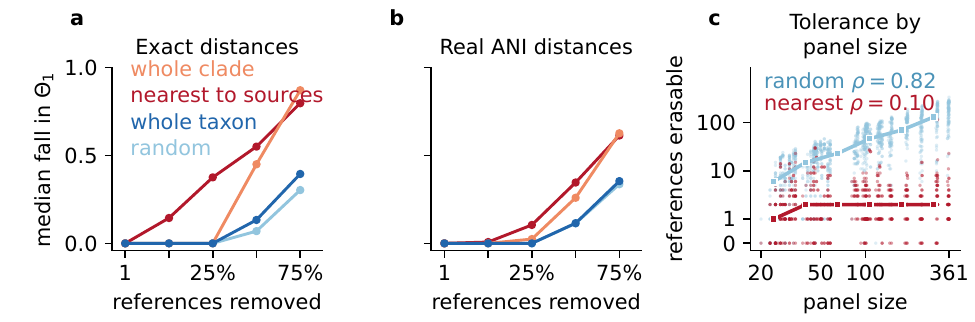}
  \caption{What incomplete references cost; in (a,b) higher is more damage. (a) Median fall in $\Theta_1$ against the share of references removed, over 296 exactly explainable profiles: 0 leaves the one-source tolerance unchanged, 1 collapses it. (b) The same draws on real ANI estimates of the same trees and the identical leaf sets; the ordering attenuates but does not reverse. (c) References erasable before the generating support stops being the unique smallest explanation, against panel size, pooled over $K=2$--10; dots are measurements, squares medians in logarithmic bins. Draws per regime and level: 592, and 296 for nearest-to-sources; (c) pools 713 adversarial and 1{,}330 random measurements.}
  \label{fig:incompleteness}
\end{figure}

\paragraph{Uniqueness.}
A median of two references nearest the sources can be erased before the true support stops being
the unique minimum. The erasure that breaks it leaves a same-size
competing support in 671 of 713 adversarial measurements, so identifiability of the support is usually
lost before the count moves, and lowers the minimum below $K$ in the other
42. Pooled over $K$, its rank correlation with panel size is
$0.10$ against $0.82$ for random loss: an order of magnitude weaker, not absent. Random loss tolerates a median of 32.5 references and whole clades 18.

\begin{figure}[t]
  \centering
  \includegraphics[width=\linewidth]{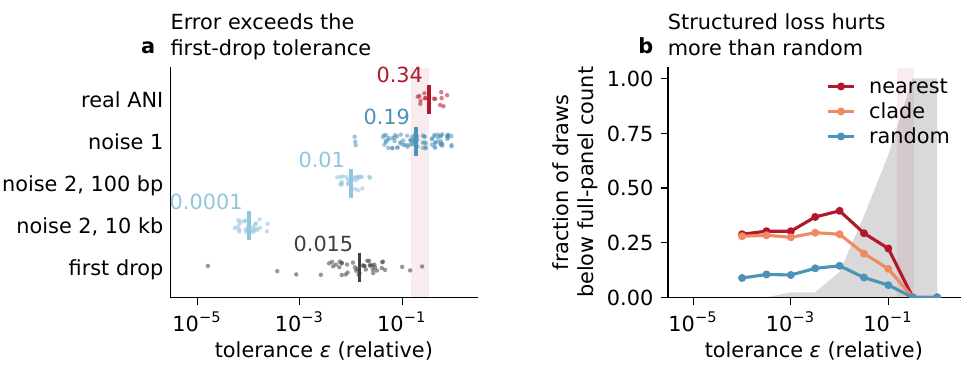}
  \caption{The tolerance axis, on exactly explainable instances; both panels scale by the case's own
  median observed distance over five decades. (a) The bottom row is the tolerance at which
  $m(\varepsilon)$ first drops below the exact source count, one point per case with a bar at the
  median; the rows above are measured distance errors on the same geometries, all at the worst leaf: two
  simulated noise models and real ANI. Real ANI error sits about an order
  of magnitude above the first drop; of the simulated settings only noise model 2 at 10\,kbp stays
  below the paired threshold at every worst leaf in all 28 cases; the band
  marks the real ANI median-to-worst-leaf range. (b) Share of removal draws whose
  $\varepsilon$-tolerant count is below their own full-panel count, by regime; the wedge is the
  share whose full-panel count has already collapsed to one source.}
  \label{fig:tolerance}
\end{figure}

\paragraph{The tolerance.}
With no reference removed, the minimum first drops below the exact count at a median tolerance of
$0.0147$ of the case's median observed distance, over 43 cases (Figure~\ref{fig:tolerance}a). It tracks the smallest constituent
of the mixture far more closely than the panel size ($\rho=0.665$ against $0.187$).
Real ANI distances deviate from their exact partners by
$0.34$ of a median observed distance at the worst leaf, exceeding the paired exact profile's first-drop
threshold in 15 of 15 cases at a median ratio of $14.5$. On that exact
profile's ladder, a tolerance holding the error leaves a minimum of one source in all 15 against three
at tolerance zero in 14. Whether the estimates themselves admit one source is a separate
computation, answered below.

\paragraph{Real distances.}
$\Theta_1$ needs no error model, no ground truth and no admissible floor, so it is measurable on
samples that have none. Computed on the estimates themselves, the worst-leaf residual of the true
configuration exceeds their own $\Theta_1$ in 296
of 296 ANI cases, median ratio 2.9, and in 50 of 60 placement cases. The sum-of-squares analogue $\Lambda_1$ agrees, 294 of 296. On the 10,575-reference backbone the contrast is wider:
discarding 7,931 references at random costs 1.8\% of $\Theta_1$, two predetermined clade deletions
of comparable size 90\%. No biological source count is inferred for any real sample.

\paragraph{One frozen heuristic release.}
Against DecoDiPhy 1.5.0 the exact method is a median 36 times cheaper end to end on the 675
exactly generated profiles it explains, while the heuristic's mixture fits real observations better
than the configuration that generated them in 282 of 296 cases, which does not establish recovery;
Supplementary Appendix~F has the timings, the budget-dependent recovery table and the norm
dependence.

%% file: sections/discussion.tex
\section{Discussion}\label{sec:discussion}
The classification separates three effects: a small source count localizes the unknown potential,
omitted subtrees require strict original-edge lifting, and fixing the mean pendant stops off-subtree
distances from being freely absorbed.

These conclusions concern exact rational data and a known metric tree.
The reductions use short edges and small margins; they do not establish
hardness at bounded aspect ratio or under fixed abundance and interior-position
floors. The prototypes are exact reference implementations, not demonstrated
replacements for large-scale heuristic inference. Arbitrarily small changes in an exact profile can raise its
minimum source count from one to any prescribed $M$, even on unit binary trees with fixed pendant
and midpoint sources (Supplementary Appendix~D), so a minimum compatible explanation need not
recover the generating composition. The tolerance is the stable quantity, $1$-Lipschitz in the
observed vector, so that perturbation moves it by at most its own size. What a profile does
support is the tolerance at which an explanation of a given size exists, and that is what methods
should report.

\paragraph{What this means for a practitioner.}
$\Theta_1$ needs no error model and no ground truth, and a fitted mixture's residual against it has
a definite reading. For a verified strict $K$-source fit with worst-leaf residual
$r_{\mathrm{fit}}$ on the same tree and panel, $r_{\mathrm{fit}}<\Theta_1$ gives
$2\le m(\varepsilon)\le K$ at every budget $r_{\mathrm{fit}}<\varepsilon<\Theta_1$: the fit is the
upper witness and no single source is admissible. On one ANI case a strict three-source witness with
$r_{\mathrm{fit}}=0.0280$ against $\Theta_1=0.0357$ gives $2\le m(0.03)\le3$, the one-source
exclusion checkable by hand from three references (Supplementary Appendix~F). These conditional model counts do not identify biological composition,
and the budget is not an uncertainty estimate. The second thing worth checking is which
references are absent rather than how many, worth a factor of fifty on the 10,575-leaf stool backbone. Code and data:
\url{https://github.com/Haotian-Tang/recomb2027-pdd}.